\documentclass[10pt, conference, letterpaper]{IEEEtran}
\IEEEoverridecommandlockouts
\usepackage{cite}
\usepackage{amsmath,amssymb,amsfonts}

\usepackage{algorithm}
\usepackage{comment}
\usepackage{algpseudocode}

\usepackage{amsthm}
\newtheorem{theorem}{Theorem}
\newtheorem{lemma}{Lemma}
\newtheorem{definition}{Definition}
\newtheorem{corollary}{Corollary}

\usepackage{color}
\usepackage[normalem]{ulem}

\newcommand{\mv}[1]{\textcolor{blue}{\ifcomments MV: #1 \fi}}

\newif\ifcomments
\commentstrue
\newcommand{\fl}[1]{}

\newcommand{\cut}[1]{}

\usepackage{xspace}
\newcommand{\protocol}{\textsc{Orion}\xspace}

\usepackage{graphicx}
\usepackage{textcomp}
\usepackage{xcolor}
\usepackage{flushend}

\usepackage{todonotes}

\graphicspath{{figures/}}

\definecolor{orange}{RGB}{225,128,0}
\definecolor{brown}{RGB}{225,128,128}
\xdefinecolor{DarkGreen}{cmyk}{0.8,0,0.8,0.3}
\definecolor{rgviolet}{RGB}{119,51,255}
\definecolor{redviolet}{RGB}{255,0,85}
\definecolor{jdbrown}{RGB}{77,38,0}

\usepackage[colorlinks,
citecolor=purple,
urlcolor=purple,
linkcolor=purple,
bookmarks=false,
hypertexnames=true]{hyperref}

\usepackage{listings}

\newboolean{showcomments}
\setboolean{showcomments}{true}

\ifthenelse{\boolean{showcomments}}
           { \newcommand{\mynote}[3]{
               \fbox{\bfseries\sffamily\scriptsize#1}
                    {\small$\blacktriangleright$\textsf{\emph{\color{#3}{#2}}}$\blacktriangleleft$}}}
           { \newcommand{\mynote}[3]{}}

\definecolor{orange}{RGB}{225,128,0}
\definecolor{brown}{RGB}{225,128,128}
\definecolor{purple}{rgb}{0.54, 0.17, 0.89}
\definecolor{rggold}{RGB}{218,165,32}
\definecolor{rgviolet}{RGB}{119,51,255}

\newcommand{\stb}[1]{\mynote{Story:}{#1}{red}} %STory

\ifthenelse{\boolean{showcomments}}
           { \newcommand{\textline}{\begin{center}\color{red}\hrulefill\end{center}}
           }{\newcommand{\textline}{}}

\usepackage{lscape}

\begin{document}

\title{
  %Leveraging Hierarchical Communication in Blockchain \\
  %Compositional Construction of a Hierarchical Geo-replicated Consensus Protocol \\
  Tolerating Disasters with Hierarchical Consensus
  \thanks{This work is supported through FNR/FCT grant C18/IS/12694392 (ThreatAdapt). For the purpose of open access, and in fulfilment of the obligations arising from the grant agreement, the author has applied a Creative Commons Attribution 4.0 International (CC BY 4.0) license to any Author Accepted Manuscript version arising from this submission.}}

% \author{
%    \IEEEauthorblockN{Wassim Yahyaoui$^1$, Joachim Bruneau-Queyreix$^2$, Marcus Völp$^1$, Jérémie Decouchant$^3$}
%    \IEEEauthorblockA{
%    $^1$ University of Luxembourg  \qquad $^3$ TU Delft \quad $^2$ Enseirb-Matmeca\\
%      SnT - CritiX group \qquad Delft Blockchain Lab \quad LaBRI lab\\
%      \url{wassim.yahyaoui@uni.lu}, 
%      \url{j.decouchant@tudelft.nl},
%      \url{joachim.bruneau-queyreix@u-bordeaux.fr}, \url{marcus.voelp@uni.lu} 
%      }
% }

\author{
  \IEEEauthorblockN{Wassim Yahyaoui\IEEEauthorrefmark{1}, Joachim Bruneau-Queyreix\IEEEauthorrefmark{2}, Marcus Völp\IEEEauthorrefmark{1}, Jérémie Decouchant\IEEEauthorrefmark{3}}
  \IEEEauthorblockA{
    \IEEEauthorrefmark{1}SnT - University of Luxembourg,
    \IEEEauthorrefmark{2}CNRS-LaBRI - Univ. Bordeaux - Bordeaux INP,
    \IEEEauthorrefmark{3}Delft University of Technology \\
    wassim.yahyaoui@uni.lu, joachim.bruneau-queyreix@u-bordeaux.fr, marcus.voelp@uni.lu, j.decouchant@tudelft.nl}
}

\maketitle

\begin{abstract}

  %\mv{TODO: 
  %	\begin{itemize}
  %		\item replace global cluster with global group
  %		\item replace node with replica
  %	\end{itemize}}

  Geo-replication provides disaster recovery after catastrophic accidental failures or attacks, such as fires, blackouts or denial-of-service attacks to a data center or region.
  Naturally distributed data structures, such as Blockchains, when well designed, are immune against such disruptions, but they also benefit from leveraging locality.
  In this work, we consolidate the performance of geo-replicated consensus by leveraging novel insights about hierarchical consensus and a construction methodology that allows creating novel protocols from existing building blocks.
  In particular we show that cluster confirmation, paired with subgroup rotation, allows protocols to safely operate through situations where all members of the global consensus group are Byzantine.
  We demonstrate our compositional construction by combining the recent HotStuff and Damysus protocols into a hierarchical geo-replicated blockchain with global durability guarantees.
  We present a compositionality proof and demonstrate the correctness of our protocol, including its ability to tolerate cluster crashes. Our protocol --- \protocol\footnote{In the Greek mythology, Orion is a blind Giant that carried his servant on his shoulders to see for him. Our protocol uses cluster-confirmation to lift Damysus to the global level and use it for global consensus.} --- achieves a 20\%  higher throughput than GeoBFT, the latest hierarchical Byzantine Fault-Tolerant (BFT) protocol.

\end{abstract}


\section{Introduction}
\label{sec:introduction}
%==================== + abstract = 1 page

Geo-replication provides disaster recovery by storing important information geographically distributed to tolerate entire sites crashing or becoming partitioned from the network. In addition, cyberattacks and similar accidental or intentionally malicious incidents may compromise individual replicas at a site, causing replicas to act irrationally or possibly even maliciously.
To compensate, sites replicate nodes and run Byzantine fault-tolerant consensus protocols to reach agreement about the transactions they commit.
In a geo-replicated setting, such protocols are hierarchical with the low-layer groups --- called \emph{clusters} --- aiming for consensus among the replicas of the same data center or region and a high-layer replica group --- the \emph{global group}
--- striving for agreement among data centers to ensure consistency across locations.

%\mv{Add me where I belong: ``and composed of representatives of the low-layer replicas groups''}
%\mv{Added to cluster confirmation, since we use representatives shortly after and for the first time}

Several hierarchical consensus protocols have been proposed~\cite{berger2023sok,amir2008steward,thai2019hierarchical,gupta2020resilientdb,jiang2023scalable,zhang2018research,feng2018scalable,qushtom2022high,li2021optimized}. However, these proposals grapple with high-cost view-change protocols, often resulting in suboptimal performance and hindering practical application, albeit with limited decentralization or consistency.
For example, Steward~\cite{amir2008steward} ensures global ordering and durability by considering transactions of a single cluster at a time, GeoBFT~\cite{gupta2020resilientdb} ensures only durability, but not global ordering, which prevents clients from interacting with information that is not maintained by their local cluster, and HiBFT~\cite{thai2019hierarchical} only orders globally, but not locally, which prevents local ahead-of-commit conflict resolution for the majority of transactions that
%are constrained to a cluster's
only affect exclusively locally modified objects.

In this paper we present \protocol, a hierarchical consensus protocol that, in contrast to previous works, harnesses the full potential of cluster parallelism and that
%preemptively
proactively orders requests locally to expedite conflict resolution and global ordering, offering clients the flexibility to interact with all objects as necessary.
It accomplishes this while effectively circumventing the pitfalls of expensive view-change protocols that limit its predecessors.

%Furthermore, 
\protocol distinguishes itself from previous protocols designed specifically for hierarchical consensus. This work leverages a novel insight in the construction of BFT protocols that was originally introduced in Steward: the construction of hierarchical protocols from non-hierarchical building blocks.
More precisely, we found that cluster consensus can not only constrain Byzantine replicas to behave like crashed ones, but in fact cluster consensus can replace trusted-trustworthy components as they are found in hybrid BFT protocols, such as MinBFT~\cite{veronese2011efficient} or Damysus~\cite{decouchant2022damysus}, and used there to reduce the required number of replicas and protocol steps\footnote{See also Gupta et al.~\cite{10.1145/3552326.3587455} and Bessani et al.~\cite{bessani2023vivisecting} for a recent discussion about hybrid BFT protocols.}.
Cluster confirmation --- i.e., the confirmation of all actions of global-group replicas (or, in our case, of replicas that originate from and represent the local cluster) through a majority of other local replicas --- thereby prevents Byzantine representatives from jeopardizing the consistency of the blockchain, while a cluster-global rotation scheme ensures liveness, even if temporarily all representatives in the global group are Byzantine.
%
%Cluster confirmation, paired with a rotation scheme that was originally introduced in HotStuff~\cite{yin2019hotstuff} and that replaces the representatives after each global consensus round and if needed after a global timeout, allows us to avoid the complicated and costly view-change constructions on which the above protocols rely to maintain liveness, even if temporarily all representatives of the global group are Byzantine.

\protocol exemplifies our generic construction method by leveraging as building blocks a consistent broadcast protocol, HotStuff~\cite{yin2019hotstuff} for early request-preordering and local conflict resolution, and Damysus~\cite{decouchant2022damysus} as hybrid protocol for reaching consensus in the global group.
We achieve the latter by implementing Damysus' trusted-component services --- checker and accumulator --- through cluster confirmation.
\protocol achieves 20\% higher throughput than GeoBFT at a slight increase in latency, measured in the time until the client receives confirmation that its requests are durably stored and in a manner that is robust to whole cluster crashes.

Specifically, \protocol features the following key points.
\begin{itemize}
  \item \protocol leverages a local pre-ordering for conflict resolution locally at a cluster without global interaction

  \item It is the first hierarchical protocol based on HotStuff, which avoids complex and costly view changes

  \item \protocol combines a consistent broadcast protocol for transaction dissemination before global ordering, and an inter-cluster SMR protocol that utilizes the Damysus hybrid BFT-SMR protocol to reduce the required number of clusters and global communication phases.

  \item A key achievement of \protocol is its novel compositionality, effectively integrating various building blocks into a cohesive framework. This unique assembly not only achieves scalability and high throughput but also significant enhances the efficiency and robustness of Byzantine fault-tolerant consensus protocols.
\end{itemize}

%% old version
%\begin{itemize}
%\item Local pre-ordering to resolve conflicts on a cluster's \emph{local objects} without global interaction. Local objects are those accessed exclusively by the clients that are assigned to that cluster. We allow for cross-cluster transactions and we keep all loca l objects globally replicated to ensure durability and to provide clients access to them through other clusters in case their local cluster crashes;
%\item Our protocol is the first hierarchical protocol that is based on Hotstuff, which allows us to avoid complicated and costly view changes;
%\item We leverage consistent broadcast protocols to disseminate transactions before ordering them globally; and 
%\item We leverage the hybrid BFT-SMR protocol Damysus to reduce the number of clusters required to tolerate crashes and the number of global communication phases.
%\end{itemize}

%\begin{itemize}
%\item 
%The rest of this paper is organized as follows.
After discussing related work (in Sec.~\ref{sec:related_work}) and our fault and system models (in Sec.~\ref{sec:threat_model}),
%, we introduce in Sec.~\ref{sec:threat_model} our fault and system models.
% which extends previous works by allowing up to $f$ replicas of each cluster to fail in an arbitrary (i.e., Byzantine) manner. In addition up to $F$ clusters may fail in the sense that additional replicas crash, become disconnected or fail silently, which might cause the complete cluster to crash (e.g., in case of disaster). 
%\item 
we introduce (in Sec.~\ref{sec:approach}) \protocol, our approach to geo-replicated Byzantine fault-tolerant state machine replication (BFT-SMR) and Blockchains.
%and the interplay of the components we leverage.
Sec.~\ref{sec:details} describes the \protocol steps in detail.
Our construction from different base protocols is made possible by a novel compositionality result, which we prove in Sec.~\ref{sec:proofs}.
%use to prove \protocol in Sec.~\ref{sec:proofs}. 
Sec.~\ref{sec:evaluation} evaluates our protocol and compares it to GeoBFT and non-hierarchical baselines. Sec.~\ref{sec:conclusion} concludes this paper.

\section{Related Work}
\label{sec:related_work}
%========================= 1 page
%\jbq{make it one page}

\textbf{Classical and streamlined BFT consensus.}
Classical BFT consensus algorithms assume partially synchronous networks, tolerate up to $f$ Byzantine replicas and require at least $3f+1$ replicas to guarantee both safety and liveness. These algorithms include the seminal PBFT~\cite{castro2002practical}, BFT-SMaRt~\cite{bessani2014state}, Zyzzyva~\cite{kotla2007zyzzyva}, are leader-based and rely on a view-change scheme.
The quadratic message complexity of these protocols typically result in low performance in geo-distributed settings.
Improving over PBFT's performance, HotStuff~\cite{yin2019hotstuff} avoids all-to-all communication patterns, by
integrating view-change procedure into the steady case, leading to linear message complexity and increased throughput.
In return Hotstuff requires one additional communication phase.
In this work, we use Hotstuff pre-order requests within a cluster and extend its rotation scheme to avoid view changes, locally and globally.

\textbf{Trusted components.}
Another line of work aimed at leveraging trusted components inside BFT consensus algorithms to reduce the number of replicas they require, and possibly the number of communication phases they employ.
Possible trusted components have included trusted logs~\cite{haeberlen2007peerreview}, attested append-only memory
(A2M)~\cite{Chun+Maniatis+Shenker+Kubiatowicz:sosp:2007}, and multiple trusted counters~\cite{Levin+Douceur+Lorch+Moscibroda:usenix:2009}.
Following this line of work, MinBFT~\cite{veronese2011efficient} builds on PBFT and uses a trusted monotonic counter to make omissions and equivocation detectable and reduce the number of replicas to $2f+1$ and the communication phases to two.
%counter, which is hosted by a trusted component, so that the leader assigns a monotonically increasing sequence number to its proposals. Consequently, MinBFT uses $2f+1$ replicas and 2 communication phases.
% MinBFT \cite{veronese2011efficient}:The main idea behind MinBFT is to reduce the number of communication steps which is considered as the metric for latency in BFT algorithms. To do so, the authors came up with a solution to reduce the number of replicas from $3f+1$ to $2f+1$ by equipping the server with a trusted component (tamperproof). Specifically, the idea is to use a trusted counter to assign sequence numbers to client requests by the primary. Added to that, the component should provide a signed certificate to prove that the number is assigned only to that message and not others together with incrementing the counter thus the same number won’t be used twice.
% In order to improve robustness and decrease communication rounds, as in hybrid protocols, 
Damysus~\cite{decouchant2022damysus} identifies two trusted services for streamlined BFT protocols. To ensure that replicas cannot fabricate information about the most recently prepared blocks, these services record additional block-related information. Damysus removes one communication phase (two network latencies) from HotStuff and also uses $2f+1$ replicas. Although we do not use trusted components in this work, we leverage Damysus at the global communication level. More precisely, under our system model the quorum certificate from a cluster in our protocol is assimilated as a trusted component's signature in Damysus.
%\wh{should we cite the paper "Dissecting BFT Consensus" and it's response paper from Paulo ?}
%\mv{Better as footnote in the intro}

\textbf{Hierarchical BFT.}
Organizing replicas in a hierarchical manner is another way to improve the performance of consensus in wide-area networks~\cite{berger2023sok}.
% Diverging the consensus groups into local and global entities, each one reaching consensus on the local transactions, and then primaries of each forming the global group agreeing on the final decision, different consensus algorithms  \cite{jiang2023scalable} \cite{zhang2018research} \cite{feng2018scalable} \cite{qushtom2022high} \cite{li2021optimized} have been extended to two-layer architecture using PBFT or mostly a customized version of the algorithm \cite{xu2021concurrent} on both levels of communications, others combine various protocols on the two-layer hierarchy \cite{wang2020beh} \cite{wen2020dp}. Moreover, other protocols \cite{li2020scalable, javad2021saguaro} adopts a variable number of layers. Yet, the core algorithm in both levels of all of these systems is PBFT. 
Hierarchical protocols have mostly relied on PBFT, while we leverage recent streamlined protocols, and separate block dissemination from ordering to obtain a higher throughput.
% The topic of hierarchical protocols has been well explored, and we discuss the most representative ones in the following. 
% In an era where technology is increasingly interconnected, the robustness and security of distributed systems have become paramount. Hierarchical consensus offers a pragmatic solution to the scaling problem prevalent in Byzantine Fault Tolerant (BFT) systems \cite{berger2023sok}, augmenting their efficiency without compromising their reliability. 
% By employing a layered decision-making process, this model provides resilience to failures and malicious attacks, which are ever-present threats in distributed environments.
In Steward~\cite{amir2008steward}, replicas rely on PBFT inside clusters, and clusters use a two-phase commit protocol to reach global consensus. Steward assumes that up to one-third of a cluster's replicas can be Byzantine, while only a minority of the clusters can crash.
% Mencius~\cite{barcelona2008mencius} aims to provide efficient fault-tolerance in wide-area networks but only tolerates crash faults while we consider Byzantine faults..
% Ebawa relies on a fixed leader that can be a bottleneck or a single point of failure, and it does not exploit the network topology to optimize the communication cost. Moreover, Mencius does not tolerate Byzantine faults and does not guarantee bounded delay for updates initiated by correct processes. 
% These limitations were addressed in BFT-Mencius~\cite{milosevic2013bounded}.
% \jd{missing intuition of how Mencius and EBAWA work}
HiBFT~\cite{thai2019hierarchical} uses PBFT at both the local and the global consensus levels.
% and organizes replicas in a hierarchical way, where replicas are organized into logical groups that communicate through their group leaders. 
HiBFT processes one block per global consensus round, while we allow the asynchronous dissemination of blocks and the creation of superblocks by the global consensus layer. In addition, HiBFT uses  PBFT's heavy view-change procedure, while we build on HotStuff and Damysus, which integrate the view-change procedure in the normal case.
GeoBFT~\cite{gupta2020resilientdb}
% proposes a consensus protocol called GeoBFT for geo-scale deployment, which involves numerous replicas dispersed across a large geographical region. 
% By employing a topologically aware grouping of replicas in local clusters, enabling parallelization of consensus at the local level, and reducing communication between clusters, the protocol is intended for scalability.
%operates in rounds where  all clusters are synchronized and vote is able to propose a single client request for execution. \jd{how does the protocol work? revise discussion of GeoBFT}  
% operates in rounds by grouping replicas in a single region into a local cluster, and using a three-step process to achieve consensus on client requests: local replication, global sharing, ordering and execution.
% In addition, GeoBFT 
relies on PBFT within each cluster, and uses a representative per cluster that shares a certified client request with other clusters. Eventually all client requests for a given view are received by a representative and deterministically ordered. GeoBFT employs a complex and expensive remote view-change protocol.
% Because a completely defective cluster can force other clusters to perform view-changes because it is detected, generated, and sent remotely, remote clusters have the authority to do so. That is not in our approach because rotation will put the cluster in a healthy state.
% The lack of global transaction validation is another flaw in GeoBFT. This means that after local consensus, clusters share transactions, which other clusters then execute when they receive them. Without global consensus (validation), there is no guarantee that the majority received the same transaction in this case. As a result, if there were faulty leaders during the same round, they might have shared globally distinct transactions.
% We use GeoBFT as a baseline in our performance evaluation because of its high performance and because its code is publicly available.

Other hierarchical protocols based on PBFT include~\cite{jiang2023scalable,zhang2018research,feng2018scalable,qushtom2022high,li2021optimized, xu2021concurrent}. Some protocols build on different protocols in the two-layer hierarchy~\cite{wang2020beh, wen2020dp}, while a few adopt a variable number of layers~\cite{li2020scalable, javad2021saguaro}.
% Despite the variability, the core algorithm in both levels of these systems generally aligns with PBFT.
Regardless of the unique mechanisms they employ, all these protocols grapple with the challenges associated with costly view-change operations and their throughput is limited to a single cluster proposal per round.
% This highlights that even diverse applications of PBFT, in all its forms, are not exempt from the pitfalls of expensive view-change protocols.
% Other interesting works on hierarchical protocols include~\cite{epstein2008using, jiang2019scalable, kwak2020design}.  

\textbf{Making BFT scale.}
We leverage several design paradigms that have been described in previous works. First, we separate block dissemination from their ordering, because blocks generated at the global level only contain the hashes of blocks generated by clusters. A similar mechanism was used in Narwhal/Tusk~\cite{danezis2022narwhal} and in Bullshark~\cite{spiegelman2022bullshark}, which are DAG-based~\cite{baird2016swirlds, gkagol2019aleph, sompolinsky2015secure}.
% However, contrary to Narwhal/Tusk, which use reliable broadcast, we can afford to use a lighter broadcast protocol (i.e., consistent broadcast) because our blocks carry the signature of a cluster and equivocation is not possible. 
The idea of assembling superblocks was used in HoneyBadgerBFT~\cite{miller2016honey} and RedBelly~\cite{crain2021red}.
Other methods that could be used to improve the performance of BFT protocols include the use of multiple leaders, which Mir-BFT~\cite{stathakopoulou2019mir} and Alder ~\cite{korkmaz2022alder} pioneered. These protocols are built around PBFT or Algorand~\cite{gilad2017algorand} and are therefore not directly applicable to HotStuff, which is the protocol we use. We do not employ sharding techniques~\cite{luu2016secure,zamani2018rapidchain} in \protocol as their performance depends on the actual workloads being executed but they are fully compatible with its design.

\begin{table*}[h]
  \caption{Key metrics of \protocol and related works}
  \label{tab:rw}
  \vspace{-0.35cm} - $n$: \# replicas per cluster - $f$: \# faulty replicas in each cluster - $c$: \# clusters - $F$: \# faulty clusters -
  %	Analysis of our protocol in comparison to related work. Message counts don't include local communications but self-messages. }
  \centering
  \resizebox{\textwidth}{!}{%
    \begin{tabular}{|c|c|c|c|c|c|c|}
      \hline
                                        &             & Total            & Threat model   & Msg complexity & Msg complexity & Communication steps \\
                                        & \# Clusters & \# Replicas      & cluster/global & of normal case & of view change & (+view change)      \\
      \hline
      % MirBFT\cite{stathakopoulou2019mir} & 1    & $3f+1$     & BFT        & O(mn²)  & (18$f^2$+15$f$+3)$m$     & 3(+2)\\ 
      % \hline
      % RedBelly\cite{crain2021red}  & 1  & 3$f$+1  & BFT   & O(n²)    & 9$f^2$+6$f$+1 & 3(+2) \\ \hline
      PBFT\cite{castro2002practical}    & 1           & 3$f$+1           & BFT            & O(n²)          & O(n²)          & 3(+2)               \\
      \hline
      HotStuff\cite{yin2019hotstuff}    & 1           & 3$f$+1           & BFT            & O(n)           & -              & 8                   \\
      \hline
      \hline
      Steward\cite{amir2008steward}     & 2$F$+1      & (3$f$+1)(2$F$+1) & BFT/CFT        & O(cn²+c²)      & O(c²n²)        & 3(+4)               \\
      \hline
      HiBFT\cite{thai2019hierarchical}  & $3F+1$      & n*c              & -/BFT          & O(cn+c²)       & O(c²n²)        & 3(+3)               \\
      \hline
      GeoBFT\cite{gupta2020resilientdb} & c           & (3$f$+1)c        & BFT/-          & O(cn²+c²)      & O(c²n²)        & 1(+3)               \\
      \hline
      \hline
      \protocol                         & 2$F$+1      & (3$f$+1)(2$F$+1) & BFT/*CFT       & O(cn+c²+c)     & -              & 6                   \\
      (this work)                       &             &                  &                &                &                &                     \\
      \hline
    \end{tabular}
  }
\end{table*}

Table~\ref{tab:rw} provides an overview of the BFT consensus protocols that are most closely related to this work, and compares them with \protocol, our protocol. Note that \protocol's threat model is close to Steward's but contains some noticeable adjustments to more effectively manage potential faulty behaviors at the global level. In both Steward and \protocol, clusters may crash, however in \protocol, up to $F$ out of the $N=2F+1$ representatives that form the global group can be Byzantine.

\section{System and Fault Model}
\label{sec:threat_model}
%================ half page

\textbf{Clusters and Replicas.}
We consider a static system that implements a consensus service, used for example by a permissioned blockchain, in geo-replicated settings.
The system consists of $N$ clusters $C_1, \ldots C_N$.
Each cluster $C_i$ contains $n_i = 3f_i + 1$ replicas of which up to $f_i$ replicas can be arbitrarily faulty (i.e., Byzantine).
We write $r_{i,k}$ for the $k^{th}$ replica in cluster $C_i$ and omit indices where they are clear from the context.
Clusters may crash, e.g., in case of fire in a cluster, regional blackouts or network partition.
In that case, we say that a cluster has \emph{crashed}.
Our goal is to tolerate up to $f_i$ Byzantine replicas in each cluster and up to $F$ out of the $N=2F+1$ crashed clusters.

%Clusters are geographically distributed but the replicas forming a cluster and as such the cluster as a whole may be affected by regional effects, such as a fire in the cluster's data center or a regional power blackout. 
%In this regard, clusters as a whole may crash, which disrupts the cluster's service, but no more than $f_i$ of each cluster's replicas can turn Byzantine. We assume no more than $F$ clusters crash.

%\mv{Say: N = 2F+1}  

Within each cluster, one replica is a local leader (denoted $r^i_L$) and is responsible for pre-ordering requests locally.
Leaders are rotated at each view change, which is a global event.
Replicas are able to identify who leads a cluster based on the global view $v$ (e.g., by choosing as leader $r^i_L$ of cluster $C_i$ the $L^{th}$ replica where $L = v\ \mathit{mod}\ n_i$).

%For each cluster, we distinguish a leader (denoted $r^i_L$ or $r_L$), responsible for steering the cluster's local operation. We rotate the leader at every view change (which is a global event) such that for each global view $v$, all replicas know the distinct leader of each cluster (e.g., by choosing $r^i_L$ in view $v$ to be the $L^{th}$ replica of $C_i$, where $L = v\  \mathit{mod}\ n_i$).

\textbf{Clients.}
Clients interact primarily with the cluster $C_i$ that is closest to them, but will be able to receive service from other clusters should $C_i$ crash.
We call $C_i$ the \emph{local cluster} of such a client.
Any number of clients may be faulty.

%use and interact with the service or blockchain implemented by our geo-replicated service. Each client $cl_j$ has a preferred cluster $C_i$ (typically the one closest to it), but will be able to request services from other clusters should $C_i$ crash. Our goal is to ensure durability of transactions despite up to $F$ cluster crashes. We make no assumption on the number of faulty clients.

\textbf{Global Group.}
In addition to the cluster leaders, we distinguish a second replica in each cluster --- the \emph{representative} $R^i$. Representatives form the global group.
Leader and representative might be the same replica, but, for load-balancing reasons, choosing different replicas suggests itself.
We achieve this by rotating the representative on global view changes and by adjusting local rotation
%the rotation scheme of the local protocol
to skip over selected representatives when selecting the next local leader.

Notice that because representatives are cluster replicas, some might belong in crashed clusters. All members of the global group can be  Byzantine. We tolerate these obstacles by rotating out of these configurations and eventually reaching one where at most $F$ representatives originate from crashed clusters or are Byzantine.

\iffalse
  Replicas of a cluster $C_i$ communicate among each other through a
  reliable communication protocol that ensures messages are eventually
  delivered, unchanged to its destination. We call such communication
  \emph{local}. In addition, we distinguish a second dedicated replica
  for each cluster --- the cluster's \emph{representative} $R^i$ --- which
  rotates as well and which could be the leader. However, for
  load-balancing reasons we recommend selecting a different replica. In
  a view $v$ all representatives of all clusters form themselves a group
  of replicas, which we call the \emph{global group}. Like local
  communication, we assume reliable global communication via message exchange,
  using a reliable communication protocol. However, unlike local communication, we assume much higher
  latency for global communication when compared to local message
  exchange.

  Notice that the representative $R^i$ of a cluster $C_i$ might be
  Byzantine and that in certain views $v$ this might be true for an
  arbitrary number of clusters (up to $N$). As such, it is not
  guaranteed that a healthy replica exists in the current global
  cluster. It might even occur that the global group is temporarily
  (in certain views) comprised of all Byzantine replicas. We shall see
  in Section~\ref{sec:xxx} why our protocol can cope with such a
  situation.
\fi

\textbf{Communication and Synchrony.}
We assume communication among replicas of a given cluster to exhibit a much lower latency and higher bandwidth on average than inter-cluster communications.
% However, despite this tighter coupling of clusters, we decided to 
We build upon the partially synchronous model~\cite{dwork1988consensus}.
% to remain immune to time-domain attacks. \jd{time-domain attacks?} 
That is, we assume that after an unknown global stabilization time $GST$ there exists a known bound $\Delta$ such that communication latency remains below $\Delta$.
% and execution of non-faulty replicas complete within $\Delta$.

%Despite the tight coupling of replicas in a cluster, we decided to build our work on a partially synchronous model to remain immune to time-domain attacks. That is, we assume after an unknown Global Stabilization Time (GST), there is a known bound $\Delta$ such that all communication and execution of healthy replicas complete within $\Delta$.

\textbf{Cryptographic Primitives.}
We assume the cryptographic functions we use to be secure, e.g., signatures cannot be forged and hashing schemes are collision-resistant.
% and threshold signatures to be unforgeable even if any number of parts below the threshold become known. 
Moreover, we assume replicas to be deployed securely (e.g., using authenticated boot) and healthy ones to not leak credentials.

\section{%Towards 
  Geo-Replicated BFT-SMR and Blockchains}
\label{sec:approach}
%============== half page

We first give an overview of our geo-replicated permissioned blockchain, and of \protocol, the hierarchical Byzantine-fault-tolerant consensus protocol it uses.

%The following section gives an overview of our geo-replicated permissioned blockchain, and \protocol, the hierarchical Byzantine-fault-tolerant consensus protocol it uses. 

\begin{figure}[b]
  \centering
  \includegraphics[width=0.7\linewidth]{ 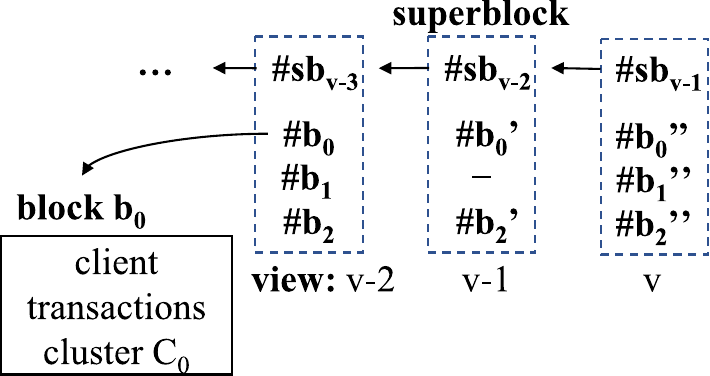}
  \caption{Structure of our blockchain.}
  \label{fig:blockchain}
\end{figure}

In the absence of faults, clients communicate with their preferred cluster to issue transactions that should be made durable by placing them in the geo-replicated blockchain before executing them.
Clusters assemble client transactions into blocks, disseminate them for other clusters to  store them durably, order them into the local chain, and request the global group to include them into the blockchain.
The global group in turn assembles received blocks into a superblock and appends it to the blockchain.
% , while preserving their local order, of the different clusters into a superblock and 
Once appended, clusters resolve global conflicts in a deterministic manner and, if required, execute the contained transactions to update their state. Fig.~\ref{fig:blockchain} shows the structure of our blockchain.

To minimize the amount of data that needs to be communicated across clusters, we limit the communication of transactions to a block dissemination phase and reference blocks using their secure hash during the remainder of the protocol.
Consequently, superblocks store hashes to blocks plus one hash to connect to the last superblock already in the chain.
A superblock contains a maximum number $k$ of hashes of blocks proposed by the clusters and it is possible that it contains more than one block from a given cluster. The conditions for accepting a superblock are (1) that blocks from a cluster are queued in a local-order preserving manner and (2) stored durably in at least $F+1$ clusters. Healthy global leaders will further try to balance how many blocks are stored from each cluster, a property which faulty global leaders may of course try to jeopardize.
% We compensate after rotating to a healthy leader by allowing it to a balance compensating number of blocks.

% \fi

%\mv{New version:}
%We store one block from each cluster or the special hash-value --- in case no block could be stored. 
%The conditions for accepting a superblock are (1) that blocks from a cluster are queued in a local-order preserving manner and (2) stored durably in at least $F+1$ clusters. That is, to fulfil the first condition, the hash in the superblock must refer to the last block of a local subchain, which connects blocks such that the hash of the first block refers to block chained in the previous superblock. Moreover, all blocks of this subchain must be stored durably in at least $F+1$ clusters to fulfill the second condition.

% \begin{figure}[htp]
%     \centering
%     \includegraphics[width=0.9\linewidth]{ blockchain3.pdf}
%     \caption{Structure of our blockchain.}
%     \label{fig:blockchain}
% \end{figure}

%\jd{why not possibly more than one block per cluster?}

\begin{figure}[t]
  \includegraphics[width=\linewidth]{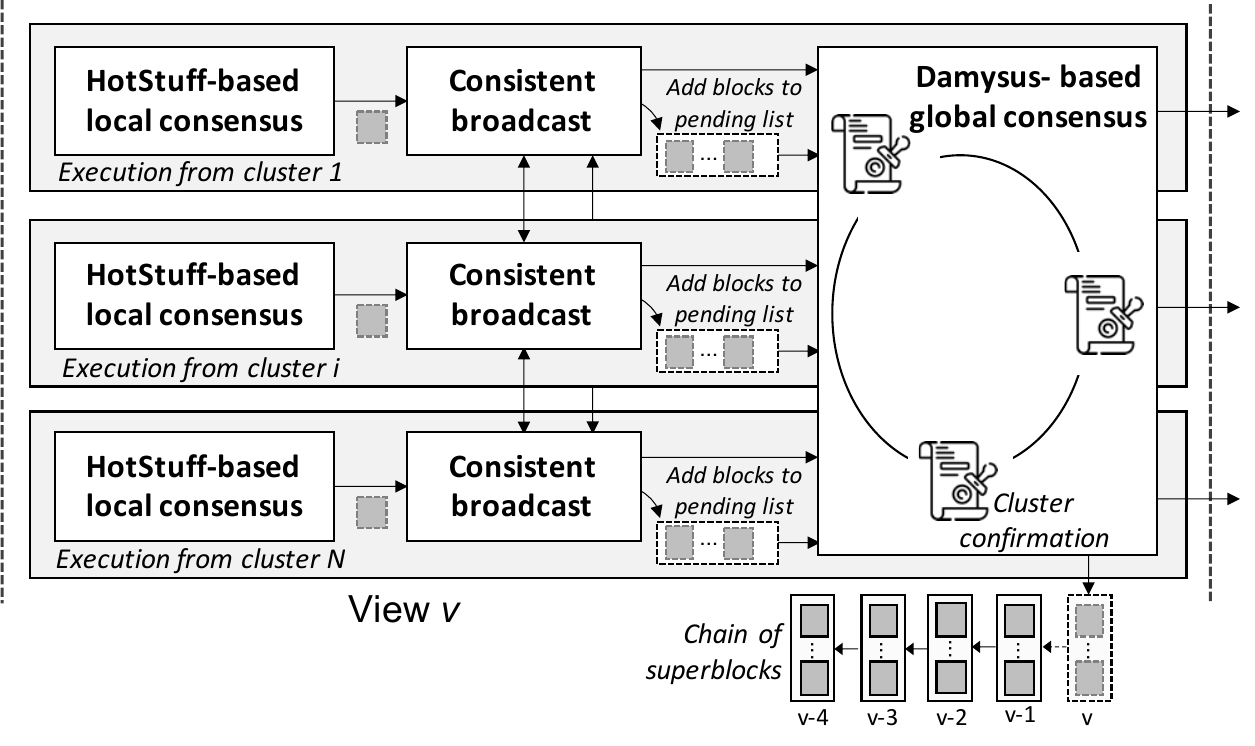}
  \caption{Overview of \protocol and the interplay of its subcomponents}
  \label{fig:protocols_interplay}
\end{figure}

\protocol combines elements from three standard protocols, which we use as building blocks: the consensus protocols Hotstuff~\cite{yin2019hotstuff} and Damysus~\cite{decouchant2022damysus}, and a consistent broadcast protocol~\cite{bracha1987asynchronous,bonomi2021practical}. Fig.~\ref{fig:protocols_interplay} depicts their interplay.
Each replica instantiates three processes:
%First, 
a \textit{local consensus protocol} based on HotStuff, which locally orders the transactions of each cluster to form a block;
%takes place within each cluster and orders transactions in the form of a block. 
a \textit{consistent broadcast protocol} that disseminates blocks to all clusters and tolerates the presence of $f$ faulty replicas in each cluster,
%and that awaits storage confirmation from $F+1$ remote clusters to ensure the block will survive up to $F$ cluster crashes; 
%Storing blocks at so many locations ensures transactions can be retrieved, even if up to $F$ clusters crash. 
and
%third, 
a \textit{global consensus protocol} based on Damysus, which forms a superblock, reaches agreement on it and includes it into the blockchain.
%responsible for superblock formation and for reaching an agreement to include this superblock into the blockchain.

% Although Figure~\ref{fig:protocols_interplay} shows the main steps of \protocol sequentially executed, they in fact operate in a pipelined and asynchronous manner with the local consensus protocol producing blocks of transactions, while the previous block is disseminated through reliable broadcast and while other blocks are appended to the chain. 
While Fig.~\ref{fig:protocols_interplay} presents the primary steps of \protocol in a sequential manner, it is crucial to note that these steps actually operate in a pipelined and asynchronous fashion.
The local protocol generates blocks of transactions while other blocks are disseminated and global agreement on superblocks with hashes of already disseminated blocks is reached.
%Meanwhile, global agreements on superblocks are being reached.
%\jd{don't get the next sentence}\mv{lets cut it} It is also important to highlight that even though only one replica per cluster actively participates in the global consensus protocol, all replicas execute the \textit{global consensus process}, to avoid startup costs during global view change. 
%This approach is adopted to avoid the costs associated with process startup and tear-down when rotating replicas as part of the global view change.

An important characteristic of our protocol is that it can handle scenarios where more than $F$ of the replicas in the global group are faulty (Byzantine or originating from crashed clusters).
In fact, the entire global group can become malicious.
To address this challenge, our protocol restricts the global group's role to simply relaying decisions made by the clusters.
To achieve this, we employ a mechanism called \emph{cluster confirmation}, which serves to gather and validate the information that replicas in the global group must relay.
Through this confirmation mechanism, equivocation attempts are immediately detected and will trigger a global view change. This effectively transforms clusters into trusted components that the global consensus protocol considers as entities that may only fail by crashing.
Upon detecting malicious behavior or on timeout, our rotation mechanism replaces the replicas in the global group until eventually a healthy configuration with less than $F$ faulty replicas is found.
%In this case, it is possible for the global group to consist entirely of crashed replicas, but the protocol's rotation mechanism will eventually transition it to a healthy configuration.
Cluster confirmation ensures that \protocol is safe and its rotation scheme eventually and repeatedly returns to healthy configurations, which ensures liveness after global stabilization.

Clusters replay disseminated blocks in case they were not included in the latest superblock. In addition, clients re-transmit transactions in case their local cluster crashes. In the latter case, we ensure re-transmitted transactions conflict with their original transaction to ensure only one becomes durable and to avoid executing the same transaction twice.

\section{Protocol Details}
\label{sec:details}
%=================== 

We now describe \protocol's sub-protocols in greater detail: block formation through local consensus, block dissemination through consistent broadcast and global agreement on a superblock through Damysus-based global consensus.

\subsection{HotStuff-based local consensus}
\label{sec:local_ordering}
%=================== half page

Exploiting the physical proximity of clients to their preferred cluster and of replicas within a cluster is key to the efficiency of hierarchical BFT-protocols.
We leverage physical proximity in two aspects: to agree on request blocks in local clusters and thereby allow for early conflict resolution of those transactions that only affect other local clients, and
by selecting a protocol --- HotStuff --- that has been optimized for fast intra-cluster conditions.
%to select a protocol --- HotStuff --- that performs optimally under reliable and fast cluster conditions and that exhibits low request-handling latencies.

%In our proposed hierarchical Byzantine fault tolerance (BFT) protocol, we utilize local replication within each cluster as a key strategy to handle client transactions in an efficient manner. Given the physical proximity, clients send their requests to the closest node, thereby minimizing transmission delays.

HotStuff is a rotating-leader-based homogeneous consensus protocol with linear communication complexity and capable of tolerating up to $f$ Byzantine replicas out of $n = 3f+1$.
% We use Basic HotStuff (in contrast to Chained) as it better aligns with our local replication strategy. 
In HotStuff, replicas agree on a single block in each view, which we then pass to Damysus for inclusion into the superblock.
Our choice of HotStuff over PBFT is driven by its higher throughput, in particular in data-center environments, as well as by its leader rotation scheme, which avoids the costs and complexities typically associated with view change, and its lower message complexity.

%\wh{to save some space the HS phases figure can be removed (doesn't add much). TODO JBQ: maybe make better figures?}

\begin{figure}[t]
  \centering
  \includegraphics[width=.95\columnwidth]{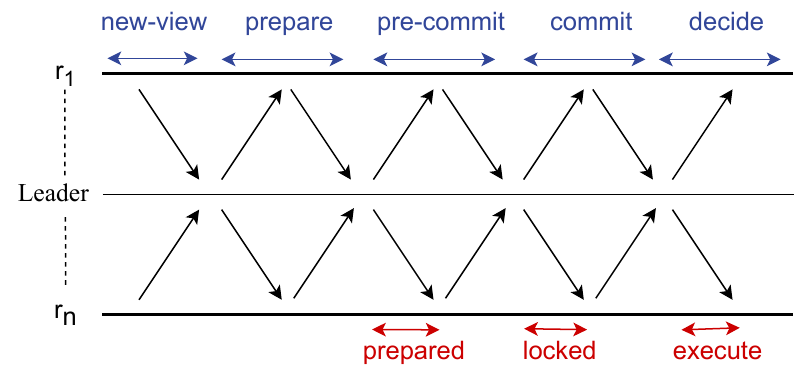}
  \caption{Communication phases in HotStuff}
  \label{fig:phs}
\end{figure}

%HotStuff comes in two main versions: Basic HotStuff and Chained HotStuff. For our use case, we employ Basic HotStuff as it aligns better with our local replication strategy. This version allows nodes to vote on a single block per view, simplifying the consensus process within each cluster.

Fig.~\ref{fig:phs} shows the protocol steps of HotStuff. After entering a \emph{new view} $v$, the rotated-in leader (here replica $R_2$) proposes a block (in \emph{prepare}) and reaches consensus (in \emph{pre-commit} and \emph{commit}). Normally the leader would then execute the \emph{decide} phase so that all replicas can execute the transaction. In \protocol, the actual execution is deferred until %after 
global consensus ensures durability.
Replicas can leverage this local pre-ordering to asynchronously (to local ordering) resolve local conflicts and disseminate blocks for later inclusion in the superblock.
%Instead, replicas may already resolve conflicts locally and disseminate the block for later inclusion in the superblock, which both happen asynchronously to the local consensus. 
%
This separation of local consensus, dissemination and global consensus enables \protocol to
proceed through these operations independently, with a positive impact on system responsiveness and scalability.
%further optimize performance by allowing these operations to proceed independently, enhancing system responsiveness and scalability.

We optimize \protocol's performance by distributing the roles of local leader and cluster representative to different replicas to avoid that a single replica has to handle coordinating consensus within a cluster and managing block execution at global level.

\subsection{Inter-cluster consistent broadcast}

\iffalse
  \begin{figure}[t]
    \centering
    \includegraphics[width=\linewidth]{global sharing fixed.pdf}
    \caption{Inter-cluster consistent broadcast}
    \label{fig:sharing}
  \end{figure}
\fi

After executing an instance of the local consensus protocol to lock a block, each correct cluster shares its local block with the other clusters.
Unlike other protocols, block dissemination is asynchronous, which separates dissemination from consensus and transaction execution.
% and relieves us from relying on timing assumptions.
%allowing us to avoid relying on any timing assumptions and separate dissemination from consensus and transaction execution.
To do so, we utilize a simple broadcast protocol to consistently disseminate local blocks between clusters and their respective local replicas.
%, which is illustrated in Fig.~\ref{fig:sharing}.
For this purpose, each cluster $C_i$ incorporates a rotating disseminator $D_{C_i}$ that initiates the broadcast of a block of transactions previously locked through the HotSuff-based consensus protocol.
The approach we employ is low-cost and effective if the disseminator is error-free.
As a result, we refrain from employing view change techniques notorious for being complicated and heavy.
The disseminator sends the block to $f+1$ replicas in each cluster, ensuring that at least one correct replica per cluster receives the block.
Once a correct replica receives the block, it relays it its peers in the local cluster.
%to all members of the local cluster.
Receiving such a block, correct replicas verify that it was properly locked in the local consensus protocol of the sending cluster before storing it in its local memory, ready to be included in the chain of superblocks.

In the event of a faulty disseminator, the rotating mechanism ensures that another healthy replica in the cluster eventually takes over the role of disseminator.
Additionally, we distribute the role of disseminator within each cluster in a way that avoids potential performance bottlenecks occurring from a single replica handling the multiple responsibilities of coordinating the local consensus, broadcasting blocks, and managing the consensus protocol at the global level.

\subsection{Damysus-based global consensus}
%========================Global ordering
%==== 2 pages including the pseudocde

Our system integrates a modified version of the Damysus\cite{decouchant2022damysus} consensus protocol aligning with wide-area network characteristics for robust blockchain replication.
The inter-cluster global consensus protocol adopts a three-phase communication structure where local replicas communicate with cluster representatives, which is illustrated in Fig.~\ref{fig:galaxy}.

The global consensus layer comprises a shared overlay network involving representatives from $N$ clusters, which participate in global consensus even when $F$ clusters are faulty. A clear distinction is made between local leaders for local consensus and representatives for global consensus, balancing replica load and minimizing faulty replica impact.

This consensus protocol is tailored for slow connectivity and high latency, enabling representatives to reach consensus on references (hashes) rather than transmitting large amounts of transaction data or whole blocks. Coupled with asynchronous block dissemination, this design optimizes consensus by ordering fixed-size references, ensuring independence of overall and consensus throughput, while totally ordering blocks.
%is not reliant on consensus throughput, yet providing total order for blocks.

\begin{figure}[t]
  \centering
  \includegraphics[width=\linewidth]{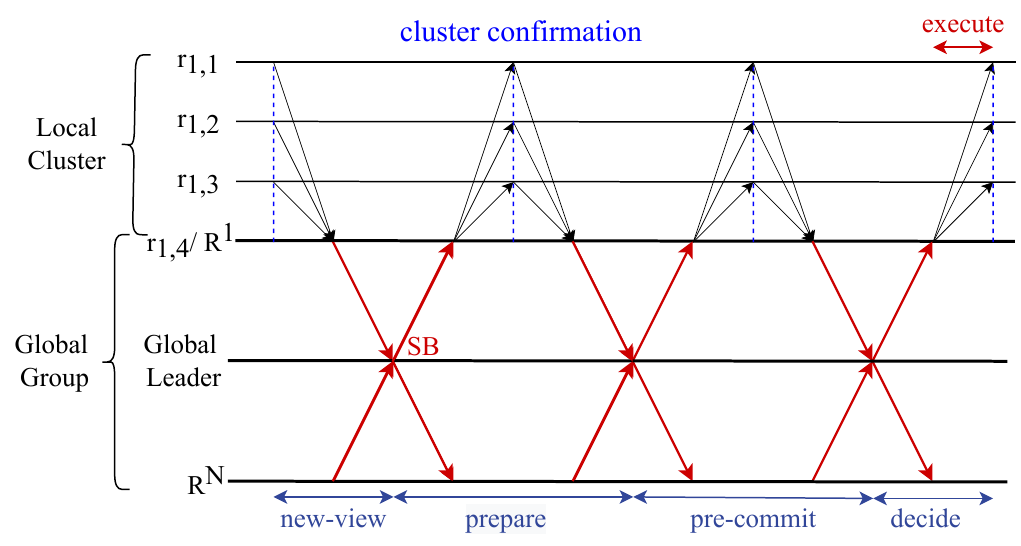}
  \caption{Phases of the Damysus-based global consensus protocol}
  \label{fig:galaxy}
\end{figure}

Unlike some hierarchical schemes~\cite{qushtom2022high}, our global layer eliminates competition for leadership in consensus rounds through a leader rotation scheme. Global leaders generate 'superblock' proposals containing references to local blocks, resulting in a lightweight protocol. Instead of carrying full block data, PREPARE messages from the global leader contain metadata including cluster-ID, replica-ID, ViewNumber, IDs of blocks, and the hash of the preceding superblock.

\textbf{Clusters confirmations.} A cluster confirmation, denoted as $\phi$ and produced by clusters, can be expressed as $\langle h, v, h', v', \text{ph}, \vec{\sigma}_{\text{c}}^{n} \rangle$. This includes the hash values of superblocks represented by $h$ and $h'$, the view numbers depicted by $v$ and $v'$, a phase represented by $ph$ can take on values from the set $\{\text{nv\_p}, \text{prep\_p}, \text{pcom\_p}\}$ (correspond to the new-view, prepare, and pre-commit phases), and a collection of signatures on the data $(h, v, h', v', \text{ph})$, indicated by $\vec{\sigma}_{\text{c}}^{n}$ = [$\sigma_{\text{c1}}$,...,$\sigma_{\text{cn}}$].

A cluster confirmation is referred to as an n-cluster confirmation if it encompasses $N$ signatures from separate clusters.
We characterize \textcolor{purple}{\hyperref[alg:cc]{C-combine}} as a function dedicated to the generation of quorum certificates from the votes of various clusters, also known as confirmations. On the other hand, \textcolor{purple}{\hyperref[alg:cc]{C-match}} validates whether a sufficient number of messages have been received from different clusters during a specific phase.

\textbf{Invoked functions.} \textcolor{purple}{\hyperref[alg:clprepare]{Cl-Prepare}}: This function generates a cluster confirmation as an approval or preparation vote based on the hash value of the proposed superblock from the global leader and an extension to the latter. The extension, proposed by the leader cluster, comes with a cluster signature certifying it as the highest superblock. This function only operates when the extension is generated by the current leader cluster.
\textcolor{purple}{\hyperref[alg:clpcom]{Cl-Pcom}}: Accepts confirmations from $F+1$ clusters, checks that a quorum of clusters prepared the superblock in the current view, and produces a cluster confirmation signifying an approval or pre-commit vote. \label{alg:cc}
\textcolor{purple}{\hyperref[alg:createclustersign]{CreateClusterSign}}: Consists of three steps: First, the cluster leader broadcasts the message to local replicas. Second, local replicas compute a partial signature on the message. Finally, the leader combines these into a cluster certificate, proving that $2fi+1$ replicas have given their approval.
\textcolor{purple}{\hyperref[alg:extlist]{ExtList}}: Selects the new view message for the superblock prepared at the highest view among $F+1$ clusters. It then certifies it as the highest by creating a cluster signature on it.
\textcolor{purple}{\hyperref[alg:clusterstart]{ClusterStart}}: After successfully verifying the signatures of cluster confirmations from remote clusters, it sets the $Ext$ variable. The confirmation is then converted into $Ext$ as a qualified potential extension for the highest superblock.
\textcolor{purple}{\hyperref[alg:clusteriterate]{ClusterIterate}}: Checks the remaining cluster confirmations for a potential higher superblock and updates $Ext$ accordingly.
\textcolor{purple}{\hyperref[alg:clusterfinalize]{ClusterFinalize}}: Takes a potential extension, verifies its signatures, and then creates leader cluster signatures on it.

This lightweight global consensus protocol, encapsulated in Alg.~\ref{alg:pseudocode}, comprises Prepare, Precommit, and Decide phases, detailed in subsequent sections.

\begin{algorithm}
  \caption{\protocol pseudocode} \label{alg:pseudocode}
  \begin{algorithmic}[1]
    \newcommand{\clustsig}{\sigma_{\text{c}}}
    \footnotesize

    \State {\color{blue}{// Protocol Initialization }}
    \State \hspace*{1em} $view \gets 0$ {\color{blue}{// Current view}}
    \State \hspace*{1em} \textbf{$c-pubs$}  {\color{blue}{// Cluster public keys}}
    \State
    %=== prepare phase
    \State {\color{blue}{// Prepare phase}}
    \State \textbf{As a} global leader:
    \State \hspace*{1em} \textbf {wait} for $\vec{\phi}$  such that \textcolor{purple}{C-match} $\langle \vec{\phi}, F+1, \perp, \text{view}, \textcolor{purple}{\text{nv\_p}} \rangle$ from representatives  \label{alg:line7}
    \State \hspace*{1em} $Ext \gets$ \textcolor{purple}{ExtList} ($\phi$)
    \State \hspace*{1em} $SB \gets$ \textcolor{purple}{CreateLeaf} ($Ext.hash$, list of blocks ids)  \label{alg:line9}
    \State \hspace*{1em}  $\phi_{prep} \gets$ \textcolor{purple}{Cl-Prepare} ($H(SB)$, $Ext$) \label{alg:line10}
    \State \hspace*{1em} \textbf {send} $\langle SB, Ext, \phi_{prep}.sign \rangle$ to $f+1$ representatives in each cluster
    \State
    \State \textbf{As a} representative:
    \State \hspace*{1em} \textbf {wait} for $\langle SB, (view,v', h', F+1,\clustsig), \clustsig' \rangle$ from the global leader
    \State \hspace*{1em} $Ext \gets \langle view,v', h', F+1,\clustsig \rangle$
    \State \hspace*{1em} $\phi_{\text{prep\_p}} \gets \langle H(SB), view, h',v', \textcolor{purple}{\text{prep\_p}}, \clustsig' \rangle$
    % \State \hspace*{1em} \textbf {Abort} if $\neg (\textcolor{purple}{VERIFY}(\phi_{\text{prep\_p}})_{pubs} \wedge SB > h'$) 
    \State \hspace*{1em} \textbf {Abort} if $\neg$(\textcolor{purple}{VERIFY} ($\phi_{\text{prep\_p}})_{c-pubs}$ $\wedge$ $SB > h'$)  \label{alg:line17}
      \State \hspace*{1em} \textbf {send} $\phi' \gets$ \textcolor{purple}{Cl-Prepare}($H(SB), Ext$) to global leader \label{alg:line18}
      \State
      %=== pre-commit phase
      \State {\color{blue}{// Pre-commit phase}}
      \State \textbf{As a} global leader:
      \State \hspace*{1em} \textbf {wait} for $\vec{\phi}$ such that \textcolor{purple}{C-match} ($\vec{\phi}$, F+1, h,view,\textcolor{purple}{\text{prep\_p}}) from representatives
      \State \hspace*{1em} \textbf {send} $\phi \gets$ \textcolor{purple}{C-combine}($\vec{\phi}$) to $f+1$ representatives in each cluster
      \State
      \State \textbf{As a} representative:
      \State \hspace*{1em} \textbf {wait} for $\langle h, view, h', v', \textcolor{purple}{\text{prep\_p}}, \vec{\clustsig}^{F+1} \rangle$ from the global leader
      \State \hspace*{1em} \textbf {send} $\phi \gets$ \textcolor{purple}{Cl-Pcom }($\langle h,\text{view},h',v',\textcolor{purple}{\text{prep\_p}},\vec{\clustsig}^{F+1} \rangle$) to the global leader
      \State
      %=== decide phase
      \State {\color{blue}{// Decide phase}}
      \State \textbf{As a} global leader:
      \State \hspace*{1em} \textbf {wait} for $\vec{\phi}$ such that \textcolor{purple}{C-match} ($\vec{\phi}$,F+1, h,view,\textcolor{purple}{\text{pcom\_p}}) from representatives
      \State \hspace*{1em} \textbf {send} $\phi \gets$ \textcolor{purple}{C-combine}($\vec{\phi}$) to $f+1$ representatives in each cluster
      \State
      \State \textbf{As a} representative:
      \State \hspace*{1em} \textbf {wait} for $\langle h, view, \perp, \perp, \textcolor{purple}{\text{pcom\_p}}, \vec{\clustsig}^{F+1} \rangle$ from the global leader
      \State \hspace*{1em} \textbf{Multicast} $\langle h, view, \perp, \perp, \textcolor{purple}{\text{pcom\_p}}, \vec{\clustsig}^{F+1} \rangle$ to local replicas
      \State \hspace*{1em} \textbf{Abort} if $\neg$(\textcolor{purple}{VERIFY} $\langle h, view, \perp, \perp, \textcolor{purple}{\text{pcom\_p}}, \vec{\clustsig}^{F+1} \rangle_{c-pubs}$)
      \State \hspace*{1em} \textbf{Execute} $SB$ corresponding to $h$ and reply to clients

      \State
      \State {\color{blue}{// NewView phase}}
      \State \textbf{For all} replicas:
      \State \hspace*{1em} When executed or timeout:
      \State \hspace*{1em} $(v, ph) \gets (view, \textcolor{purple}{\text{prep\_p}})$; view++
      \State \hspace*{1em} \textbf{While} $(v,ph) \neq (view,\textcolor{purple}{\text{nv\_p}})$ \textbf{do}
      \State \hspace*{2em} $\phi \gets$ \textcolor{purple}{ClusterSign}(); $(v,ph) \gets (\phi.vprep,  \phi.phase)$
      \State \hspace*{1em} \textbf{End while}
      \State \hspace*{1em} \textbf {send} $\phi$ to view’s leader

  \end{algorithmic}
\end{algorithm}

%%\vec{\clustsig}
%%\clustsig
%%$\vec{\phi}$
%%$\phi$

\begin{algorithm}
  \caption{Invoked Functions} \label{alg:invokedfunctions}
  \begin{algorithmic}[1]
    \footnotesize
    \newcommand{\clustsig}{\sigma_{\text{c}}}

    %\State \hspace*{1em} $clid$ {\color{blue}{// the cluster id}}
    %\State

    %========================  Extension List
    \State \textbf{function} \textcolor{purple}{ExtList} ($\vec{\phi}$) \label{alg:extlist}
    \State \hspace*{1em} $\phi_0 \gets$ commitment $\phi$ $\in$ $\vec{\phi}$  with highest $\phi_{V_{\text{vjust}}}$

    \State \hspace*{1em} $Ext \gets$ \textcolor{purple}{ClusterStart} ($\phi_0$)
    \State \hspace*{1em} \textbf {for} $\phi \in \{\phi_0\}$ \textbf {do} $Ext \gets$ \textcolor{purple}{ClusterIterate} (Ext, $\phi$)
    \State \textbf {return} \textcolor{purple}{ClusterFinalize} (Ext)

    \State \hspace*{1em}
    %========================  Create Cluster Signature

    \State \textbf{function} \textcolor{purple}{CreateClusterSign} (h, h', v') \label{alg:createclustersign}
    \State \hspace*{1em} \textbf{as a leader}: request cluster signature on $(h, \text{view}, h', v', \text{phase})$ from local replicas
    \State \hspace*{1em} \textbf{as a replica}: compute $psig$ on $(h, \text{view}, h', v', \text{phase})$ \textbf{where} $psig \gets$  \textcolor{purple}{SIGN} $(h, \text{view}, h', v', \text{phase})$
    \State \hspace*{1em} $\text{phase}++$
    \State \hspace*{1em} \textbf{as a leader:} wait for $(n-f)$ $psig$ from different replicas, then build a $\clustsig$  on $(h, \text{view}, h', v', \text{phase})$ where $\clustsig \gets$ \textcolor{purple}{SIGNS} $(h, \text{view}, h', v', \text{phase})_{2f_i+1}$
    \State \hspace*{2em} \textbf{return } $\phi \gets (h, \text{view}, h', v', \text{phase}, \clustsig)$
    \State \hspace*{1em}

    %======================== Cluster Sign
    \State \textbf{function} \textcolor{purple}{ClusterSign} () \label{alg:clustersign}

    \State \hspace*{1em} \textbf {return} $\phi \gets$ \textcolor{purple}{CreateClusterSign} ($\perp$, preph, prepv)

    \State \hspace*{1em}
    %======================== Cluster Prepare 
    \State \textbf{function} \textcolor{purple}{Cl-Prepare} (h, Ext) \label{alg:clprepare}
    \textbf {where} Ext is $\langle v,v',h',F+1,\clustsig \rangle$
    \State \hspace*{1em} If \textcolor{purple}{VERIFY} $(Ext)_{\text{c-pubs}} $ $\wedge$ view $=$v $\wedge$ h $\neq \perp$ \textbf {then}
    \State \hspace*{2em} \textbf {return} $\phi \gets$ \textcolor{purple}{CreateClusterSign} (h,h’,v’)
    \State \hspace*{1em} \textbf {end if}

    \State \hspace*{1em}
    %======================== Cluster Store 
    \State \textbf{function} \textcolor{purple}{Cl-Pcom} ($\phi$) \label{alg:clpcom} \textbf {where} $\phi$ is $\langle h,v,h',v',ph, \clustsig^{F+1} \rangle$
    \State \hspace*{1em} If \textcolor{purple}{VERIFY} $(\phi)_{\text{c-pubs}}$ $\wedge$ view $=$v $\wedge$ ph $=$\textcolor{purple}{\text{prep}} \textbf {then}
    \State \hspace*{2em} $preph \gets$ h ; $preph \gets$ v ;
    \State \hspace*{2em} \textbf {return} $\phi' \gets$ \textcolor{purple}{CreateClusterSign} (h,$\perp$,$\perp$)
    \State \hspace*{1em} \textbf {end if}

    \State \hspace*{1em}
    %========================  ClusterStart 

    \State \textbf{function} \textcolor{purple}{ClusterStart} ($\phi$) \label{alg:clusterstart} \textbf{where} $\phi$ \textbf{is} $\langle \perp, v, h', v',\textcolor{purple}{\text{nv\_p}}, \clustsig \rangle$
    \State \hspace*{1em} \textbf{if} \textcolor{purple}{VERIFY} $(\phi)_{\text{c-pubs}}$ \textbf{then}
    \State \hspace*{1em} $\clustsig' \gets$ \textcolor{purple}{SIGN} $\left(v, v', h',[\clustsig.\textcolor{purple}{\text{clid}}]\right)_{2f+1}$
    \State \hspace*{1em} \textbf{return} $Ext \gets (v, v', h', [\sigma.\textcolor{purple}{\text{clid}}], \clustsig')$
    \State \hspace*{1em} \textbf{end if}

    %========================  ClusterIterate  

    \newcommand{\clustsigtwo}{\sigma_{\text{c,2}}}

    \State \hspace*{1em}
    \State \textbf{function} \textcolor{purple}{ClusterIterate} (Ext, $\phi$) \label{alg:clusteriterate} \textbf{where} Ext \textbf{is} $\langle v_1, v_1', h_1, \vec{i}, \clustsig \rangle$ \textbf {and} $\phi$ \textbf {is} $\langle \perp, v_2, h_2, \text{{\color{purple}\text{nv\_p}}}, \clustsigtwo \rangle$
    \State \hspace*{1em} \textbf{if} $\left(\begin{array}{c} v_1=v_2 \wedge v_1' \geq v_2' \wedge \clustsigtwo.\text{clid} \notin \vec{i} \land  \\ \textcolor{purple}{VERIFY}(Ext)_{\text{c-pubs}} \land \textcolor{purple}{VERIFY}(\phi)_{\text{c-pubs}} \end{array}\right)$ \textbf {then}
    \State \hspace*{1em} $\clustsig' \gets$ \textcolor{purple}{SIGN} $\left(v_1, v_1', h_1, \vec{i}@[\clustsig.\textcolor{purple}{\text{clid}}]\right)_{2fi+1}$
    \State \hspace*{1em} \textbf{return} $Ext' \gets (v_1, v_1', h_1, \vec{i}@[\sigma.\textcolor{purple}{\text{clid}}], \clustsig')$
    \State \hspace*{1em} \textbf{end if}

    %========================  Cluster Finalize 
    \State \hspace*{1em}
    \State \textbf{function} \textcolor{purple}{ClusterFinalize} (ext) \label{alg:clusterfinalize} \textbf{where} ext \textbf{is} $\langle v, v', h, \vec{i}, \clustsig \rangle$
    \State \hspace*{1em} \textbf{If} \textcolor{purple}{VERIFY} $(Ext)_{\text{c-pubs}}$ \textbf {then}
    \State \hspace*{2em} \textbf{return} $Ext \gets (v, v', h, |\vec{i}|, \clustsig)$ \textbf{where} $\clustsig \gets$ \textcolor{purple}{SIGN} $\left(v,v', h, |\vec{i}|\right)_{2f_i+1}$
    \State \hspace*{1em} \textbf{end if}

  \end{algorithmic}
\end{algorithm}

% \begin{itemize}
% \item \textbf{The global protocol algorithm}
% \end{itemize}

%-- view refers to global view

\textbf{Prepare.} %===== Prepare phase
In this phase, the global leader collects $F+1$ new-view messages from remote clusters and their representatives (Alg.~\ref{alg:pseudocode}, ln.~\hyperref[alg:line7]{\textcolor{purple}{7}}). The leader selects the new-view message corresponding to the highest-viewed superblock from all $F+1$ received messages, validating its superiority via $2f+1$ local replica confirmation.
Subsequently, the leader extends this highest superblock with a new one, referred to as $SB$ (Alg.~\ref{alg:pseudocode} ln.~\hyperref[alg:line9]{\textcolor{purple}{9}}). This superblock $SB$ is prepared using \textcolor{purple}{\hyperref[alg:clprepare]{Cl-Prepare}} (Alg.~\ref{alg:pseudocode} ln.~\hyperref[alg:line10]{\textcolor{purple}{10}}), ensuring the transition to the next phase while generating a signature for the newly proposed superblock, also known as "cluster prepare confirmation".
The global leader then disseminates the new proposal $SB$, the extension $Ext$, the cluster confirmation formed by the leader cluster to all cluster representatives. Other clusters authenticate that $SB$ extends the superblock in $Ext$ (Alg.~\ref{alg:pseudocode} ln.~\hyperref[alg:line17]{\textcolor{purple}{17}}) and prepare the superblock proposed by the global leader through validation from a quorum of inner-cluster replicas (Alg.~\ref{alg:pseudocode} ln.~\hyperref[alg:line18]{\textcolor{purple}{18}}).

\textbf{Pre-commit.}
%===== pre-commit phase
In this
%the pre-commit 
phase, the global leader assembles $F+1$ prepare cluster-confirmations from distant clusters and dispatches a prepare ($F+1$)-cluster confirmation to every cluster representative. The clusters invoke \textcolor{purple}{\hyperref[alg:clpcom]{Cl-Pcom}} to store $SB$ and secure a cluster validation that attests the new superblock $SB$ is preserved in $2f+1$ replicas. This process yields a cluster signature, formulated via $2f+1$ local partial signatures from the same phase, ensuring that a minimum of $F+1$ clusters will keep and transmit any executed superblocks in their new-view messages.
Due to the cluster confirmation functioning as Damysus's accumulator, clusters are not required to lock blocks - a mechanism employed in HotStuff during the commit phase. This locking is unnecessary in our case, as the virtual accumulator, based on reliable cluster confirmation, guarantees that the global leader only proposes superblocks that expand upon the highest prepared superblock.

\textbf{Decide.}
%===== decide phase
In the decide phase, the global leader collects $F+1$ pre-commit clusters-confirmations from remote representatives. Upon receiving these confirmations, the global leader issues a pre-commit ($F+1$)-cluster-confirmation to all cluster representatives. These representatives, in turn, broadcast the cluster-confirmation to their respective local replicas. This action initiates the execution of the superblock across all clusters, effectively achieving consensus on the ordering of the transactions contained within the superblock.

\textbf{New-view.}
%===== new-view phase - to rephrase
Each cluster utilizes timers, commencing at the inception of each view, to transition to the subsequent global view when the current one experiences delays, potentially due to a faulty representative. The new-view phase comes into play either post completion of the current view - implying successful execution of the decide phase - or upon expiration of the timer initiated at the view's start. This mechanism is activated under certain conditions, such as a representative's inability to gather sufficient votes to inaugurate a new view or local replicas' disagreement on the representative's proposal, causing a timeout and triggering the start of a new view. This phase sees clusters progress their views and provide their votes to the global leader of the new view.

%========================end of global ordering section

\section{Compositional Safety and Liveness}
\label{sec:proofs}
%==================== 1 page
%\wh{according to the comment: Safety and liveness reasoning consists of mostly verbose explanations and could be formalized more. TODO MV: revise this section}

This section establishes the safety and liveness of \protocol by proving a more general compositionality result about hierarchical consensus with cluster confirmation and global rotation.

%\textline

We assume the standard partial synchrony model and prove liveness after GST (see Section~\ref{sec:threat_model}).
More precisely, we will show that if the subprotocols executed in the clusters are live and safe, then the composition using cluster confirmation and global rotation repeatedly returns the overall system into a state where these properties extend to the global case.

We identify points of progress in the subprotocols (e.g., the dissemination of a message, reaching consensus about a client request, or appending a block of transactions) and consider subprotocols that advance to their respective next progress point, provided they remain for a given amount of time $t$ in a configuration where such progress is possible.
Our proof then establishes that (1) Byzantine configuration cannot jeopardize the state of subprotocols as a whole (i.e., they can only jeopardize the Byzantine replica's state), and (2) global rotation repeatedly returns subprotocols into sufficiently healthy configurations (with at most $f$ Byzantine replicas or up to $F$ representatives that may itself either be Byzantine or that represent crashed clusters) and remains there for at least time $t$. Then, because global rotation repeatedly returns to such healthy state and configuration pairs where progress is possible, the local properties of the protocol extend to the system as a whole as far as safety and liveness are concerned.

\begin{definition}[Local safety/liveness]
  We say clusters are \emph{locally live/safe} if the subprotocols they execute exhibit these properties.
\end{definition}

\begin{definition}[$t$-live]
  A configuration of the protocol is \emph{t-live} if given a correct overall state it is capable of advancing to the next progress point and correct overall state, provided the protocol remains in this configuration for at least time $t$.
\end{definition}

\begin{definition}[($t$-)rotation safe]
  We say a global group protocol is \emph{rotation safe} if its rotation scheme eventually and repeatedly returns to a configuration where at most $F$ representatives are faulty. It is \emph{t-rotation safe} if it remains in such a configuration for at least time $t$.
\end{definition}

Notice that rotation safety implies that the rotation cannot only be triggered by representatives, as they might be all faulty and not trigger rotation.
State transitions of a subprotocol consist of delivery of a message, its analysis and possibly the sending of further messages to trigger further state transitions.

\begin{definition}[Cluster-confirmed state transition]
  State transitions are executed under \emph{cluster confirmation} if the delivered message is accepted only if $n-f$ replicas of the represented cluster approve this message (e.g., by signing it).
\end{definition}

To obtain cluster confirmation, representatives have to first consult their clusters' replicas, before they can send messages that trigger the execution of state transitions under this confirmation.
We require protocols to advance the state they plan to change as part of obtaining such a confirmation. In particular, we require cluster confirmation to be sequenced.
%This allows us to prove:

\begin{lemma}[Cluster confirmation retains correct states]
  \label{lem:cluster_confirm}
  Exclusively executing state transitions under cluster confirmation retains correct state in healthy replicas.
\end{lemma}

\begin{proof}
  Let $s$ be the state of a healthy replica, to confirm a state-transition triggering message, representatives must provide the state to modify $s'$ and the message $m$ to trigger the transition from $s'$ to $s''$. \emph{Case} $s = s'$: In this case, the replica is able to confirm immediately whether $m$ is appropriate, by checking the conditions of the protocol. Local safety then ensures that $s''$ is safe as well. \emph{Case} $s \neq s'$: In this case, the replica accepts $s'$ only if $n-f$ replicas confirmed the transition that lead to $s'$. If so, it can approve $m$ as in the previous case. Otherwise, it refuses approval. Since $n-f > f$ approvals are required to confirm a state transition, no faulty replica can jeopardize the state held by a healthy replica.
\end{proof}

\begin{theorem}[Compositional liveness]
  If a protocol is $t$-live if at most $F$ out of $N$ of its group members are faulty, it remains $t$-live if $t$-rotation safety is guaranteed.
\end{theorem}

\begin{proof}
  After GST, a rotation safe protocol eventually reaches a configuration with at most $F$ faulty replicas and a $t$-rotation safe protocol remains in that configuration for at least a time $t$. Lemma~\ref{lem:cluster_confirm} implies that the state of the at least $N{-}F$ healthy replicas is correct and $t$-liveness ensures progress.
  Rotation safety reaches healthy configuration repeatedly and within a bounded amount of time, which ensures liveness.
\end{proof}

\begin{theorem}[Compositional safety]
  If a protocol is safe provided at most $F$ out of $N$ replicas are faulty, then it is also safe in a $t$-rotation safe setting, provided all state transitions are cluster confirmed.
\end{theorem}

\begin{proof}
  Healthy replicas in the local cluster will approve state transitions only if they have learned about all previous cluster confirmations and only if the transition meets the protocol's safety condition. Since $n-f > f$, representatives only receive confirmation if at least one healthy replica agrees. This is equivalent to a trusted component that is exclusively connected to the representative and that provides a token that cannot be forged, stating that the protocol's safety checks succeeded.

  Since Byzantine replicas cannot jeopardize the state of healthy replicas (Lemma~\ref{lem:cluster_confirm}) and since state transitions are exclusively executed under cluster confirmation, no transition can be made that would jeopardize the state of a healthy replica. Healthy replicas therefore retain the healthy state of the protocol, which ensures safety.
\end{proof}

Notice that Byzantine representatives may still alter their internal state and even reveal results to clients, but these revelations did not receive cluster confirmation and can therefore be detected by the client as possibly faulty.

There is one fundamental difference where cluster confirmation diverges from trusted components. While the latter may retain secrets (such as a symmetric key), this is only possible by means of secret-sharing~\cite{shamir1979share,vassantlal2022cobra}, since Byzantine replicas of a cluster might reveal any secrets they know.

\protocol uses Damysus and Hotstuff as subprotocols. Their safety and liveness follows from the proofs in Decouchant et al.~\cite{decouchant2022damysus} and Yin et al.~\cite{yin2019hotstuff}. The extension to $t$-liveness is easy to see given both protocols progress through a fixed number of steps to commit a block and that transmission and computation times are bounded after GST. The same holds for the broadcast protocol we use. Our protocol rotates representatives either after a block has been added (i.e., after progress is made) or after the individual replicas in the clusters time out, triggering a global viewchange and rotation of representatives. With a timeout large enough to make progress after GST, our protocol fulfills the pre-condition for $t$-rotation safety.

What remains to be seen is that our rotation scheme eventually and repeatedly selects a configuration with at most $F$ faulty representatives.
In a crashed cluster all replicas are faulty, whereas in a non-crashed cluster at most $f$ are (and at least $2f+1$ are healthy). We can therefore pessimistically assume $F$ clusters crashed.
However, since our rotation scheme iterates through all combinations of replicas from the individual clusters it will also pass configurations where it selects healthy replicas as representatives from all $N-F$ non-crashed clusters and repeatedly so. We can therefore conclude that our protocol is safe and live.

\section{Performance Evaluation}
\label{sec:evaluation}
%==================== 2 pages including plots and conclusion

We have evaluated \protocol against hierarchical (GeoBFT~\cite{gupta2020resilientdb}) and non-hierarchical (PBFT~\cite{castro2002practical}, HotStuff~\cite{yin2019hotstuff}) state-of-the-art consensus protocols. Evaluations utilized AWS EC2 replicas of different geographical regions and used a single t2.xlarge instance at each replica. We consider replicas of a region to form a cluster and associate this cluster as the preferred one to clients of that region. We distribute clients and replicas uniformly across regions.
%We report the average performance over ten repetitions, each executed for a duration of 30 views each.
We report average performance over ten repetitions, each executed for 30 views.

In all experiments, we collect $400$ transactions to form a block and use $32$ Byte hashes to link blocks and to refer to a block from a superblock. That is, the size of superblocks is $32B$ times the number of blocks it refers to plus one (to link to the previous superblock).
%\mv{I would remove the following sentence:}
%Notice that the baseline protocols need $40B$ of additional metadata per transaction, which can result in blocks of up to $15.6KB$. 
%
Experiments warm up for 60s.
%are allowed to warm-up for 60s.

\subsection{Effects of Large-Scale Geographic Deployment}

We start by investigating the effects of large-scale geographic deployment. For that, we scale the number $F$ of clusters we tolerate to crash and hence the number of clusters $N=2F+1$, while keeping their size fixed
%, by scaling up the number of clusters, while keeping their size fixed
(to $n = 3f+1 = 10$ replicas, i.e., clusters can tolerate up to $f=3$ Byzantine replicas).
Fig.~\ref{fig:throughput-C} and~\ref{fig:latency-C} show the throughput and latency of \protocol relative to the baseline protocols. The x-axis denotes the number of clusters $N$,
which we distribute to
%that the protocol tolerates (i.e., from $F=1$ to $F=5$, which translates to $N = 3, 5, 7, 9,$ and $11$ clusters). We distribute the clusters in
North Virginia, Ohio, North California, Mumbai, Singapore, Sydney, Frankfurt, Ireland, London, Central Canada, and São Paulo.

\begin{figure}
  \begin{center}
    \includegraphics[width=\columnwidth]{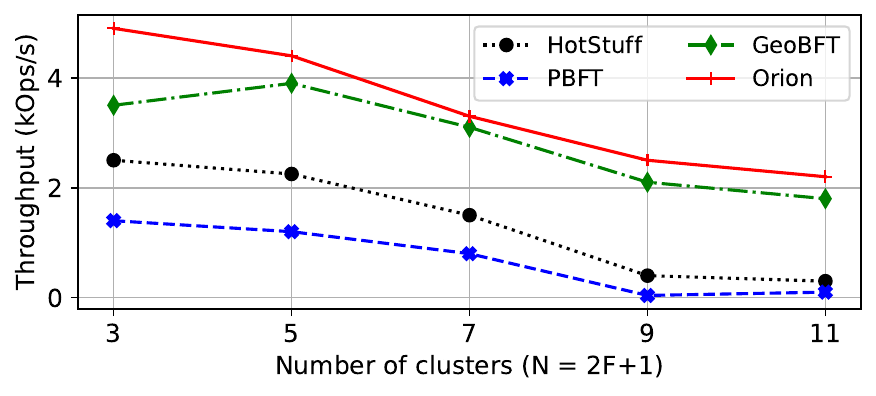}
    \caption{Throughput depending on the number $N$ of clusters (each cluster contains 10 replicas, and tolerates $f=3$ Byzantine replicas).}
    %of clusters $F$ that can be tolerated to crash.}
    %    \mv{Fix description; its not crashed clusters, its clusters that can be tolerated to crash}
    \label{fig:throughput-C}
  \end{center}
\end{figure}

\begin{figure}
  \begin{center}
    \includegraphics[width=\columnwidth]{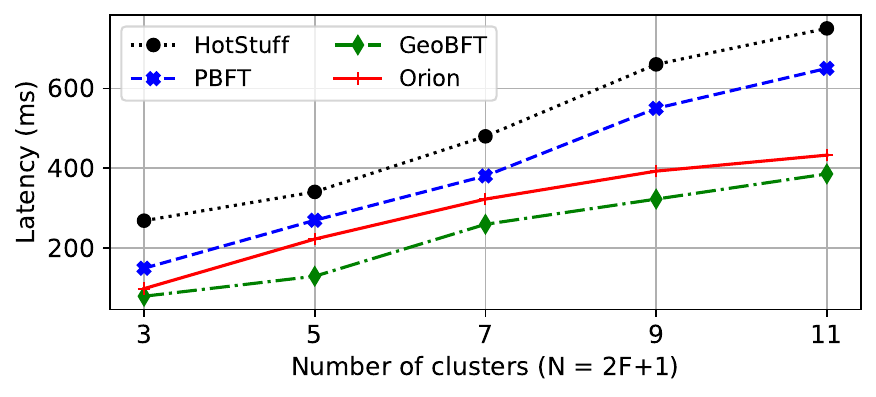}
    \caption{Latency depending on the number $N$ of clusters (each cluster contains 10 replicas, and tolerates $f=3$ Byzantine replicas).}
    \label{fig:latency-C}
  \end{center}
\end{figure}

\protocol outperforms GeoBFT's throughput (by 20\% on average), albeit at the cost of increasing the latency for handling client requests until notification after the request is durable
%request handling latency 
(by 31\% on average) and while outperforming non-hierarchical protocols in both dimensions. This throughput-for-latency tradeoff is achieved by distributing and reaching global agreement independently from the local agreement. \protocol outperforms GeoBFT in terms of throughput under faults as it relies less on inter-cluster communications.

\subsection{Influence of Increasing Local Clusters Size}

In the second part of our analysis we focus on the performance impact of increasing the size of local clusters. For these experiments, we fix the number of clusters to $N=3$ (i.e., we tolerate $F=1$ cluster crash), distributed over Ohio, Sydney, and London, while increasing the number of Byzantine replicas that each cluster can tolerate from $f=1$ to $f=5$ and cluster sizes from $n=4$ to $n=16$.
%. Consequently, local cluster sizes range from $n=4$ to $n=16$ (remember that $n = 3f+1$).  

Throughput and latency results are shown in Fig.~\ref{fig:throughput-r} and~\ref{fig:latency-r}, respectively. Again, \protocol outperforms GeoBFT in terms of throughput (by $63\%$ on average) and non-hierarchical protocols in terms of both throughput and latency. However, as local clusters grow in size, the latency advantage of GeoBFT drops to $11\%$ on average.
\protocol closes the gap to GeoBFT because GeoBFT requires global messages proportional to $f$, while this number remains constant in our protocol (independent of the cluster sizes, as only one representative communicates globally).
%This decline in performance can be attributed to the increased number of messages that are exchanged globally as clusters grow in size. In GeoBFT, this number is proportional to $f$, while it remains constant in our protocol (independently of the size of clusters, only a single representative communicates globally). 
Additionally, every cluster has to disseminate certificates to the associated $f+1$ remote replicas, which grow as well in size as $f$ increases and hence $n$ grows.
Our protocol trades performance improvements in global communication for a slight increase in local communication, since each global step requires reaching $2f+1$ local replicas for cluster confirmation, which explains the increase in latency seen in Fig.~\ref{fig:latency-r}.

\begin{figure}
  \begin{center}
    \includegraphics[width=\columnwidth]{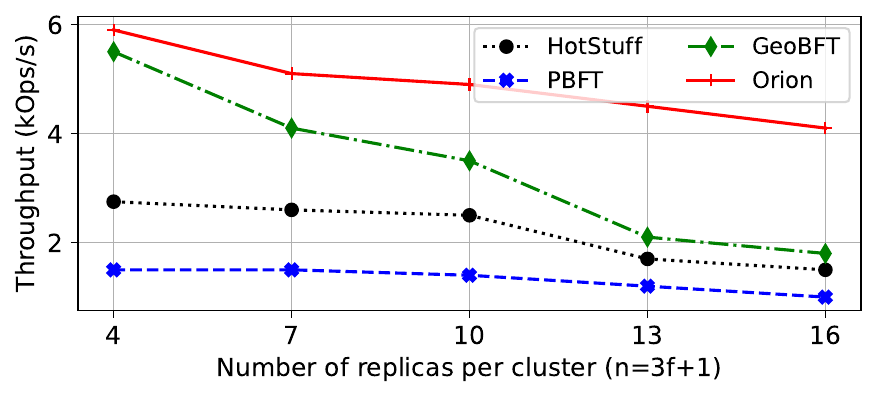}
    \caption{Throughput as a function of the number of replicas in a cluster (with $N = 3$ clusters).}
    %\mv{Fix axis}
    \label{fig:throughput-r}
  \end{center}
\end{figure}

\begin{figure}
  \begin{center}
    \includegraphics[width=\columnwidth]{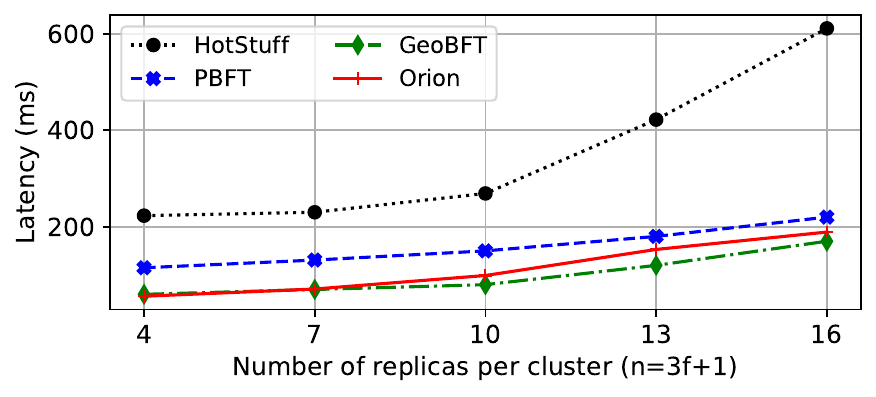}
    \caption{Latency depending on the number of replicas in a cluster (with $N=3$ clusters).}
    \label{fig:latency-r}
  \end{center}
\end{figure}

Overall, \protocol demonstrates that compositionally constructed protocols perform well and can sometimes even outperform specific hierarchical protocols. Compositionally also significantly reduces development costs.

\section{Conclusion}
\label{sec:conclusion}
%=======================

We introduced \protocol, a novel, compositionally-constructed, hierarchical Byzantine fault-tolerant consensus protocol that addresses
% marks a significant advancement in addressing development costs, while remaining competitive with 
the scalability and performance challenges of blockchain systems. We demonstrated that \protocol achieves superior throughput compared to existing protocols, like PBFT, HotStuff and GeoBFT, albeit by marginally trading off latency. \protocol is based on a novel compositionality result, enabling the use of hybrid protocols at global level without necessitating dedicated trusted components.
Looking ahead, we envisage further refinement of \protocol through expanded evaluations involving diverse clusters, scenarios, and speculative execution to improve latency. Our findings have broad implications for enhancing blockchain and distributed systems' performance, particularly in geo-distributed clustered networks, paving the way for more efficient and reliable distributed ledger technologies.

\clearpage

%======================= 1 page references
\bibliographystyle{IEEEtran}
\bibliography{biblio}

\begin{thebibliography}{10}
\providecommand{\url}[1]{#1}
\csname url@samestyle\endcsname
\providecommand{\newblock}{\relax}
\providecommand{\bibinfo}[2]{#2}
\providecommand{\BIBentrySTDinterwordspacing}{\spaceskip=0pt\relax}
\providecommand{\BIBentryALTinterwordstretchfactor}{4}
\providecommand{\BIBentryALTinterwordspacing}{\spaceskip=\fontdimen2\font plus
\BIBentryALTinterwordstretchfactor\fontdimen3\font minus
  \fontdimen4\font\relax}
\providecommand{\BIBforeignlanguage}[2]{{%
\expandafter\ifx\csname l@#1\endcsname\relax
\typeout{** WARNING: IEEEtran.bst: No hyphenation pattern has been}%
\typeout{** loaded for the language `#1'. Using the pattern for}%
\typeout{** the default language instead.}%
\else
\language=\csname l@#1\endcsname
\fi
#2}}
\providecommand{\BIBdecl}{\relax}
\BIBdecl

\bibitem{berger2023sok}
C.~Berger, S.~Schwarz-R{\"u}sch, A.~Vogel, K.~Bleeke, L.~Jehl, H.~P. Reiser,
  and R.~Kapitza, ``Sok: Scalability techniques for {BFT} consensus,''
  \emph{arXiv preprint arXiv:2303.11045}, 2023.

\bibitem{amir2008steward}
Y.~Amir, C.~Danilov, D.~Dolev, J.~Kirsch, J.~Lane, C.~Nita-Rotaru, J.~Olsen,
  and D.~Zage, ``Steward: Scaling byzantine fault-tolerant replication to wide
  area networks,'' \emph{IEEE Transactions on Dependable and Secure Computing},
  vol.~7, no.~1, pp. 80--93, 2008.

\bibitem{thai2019hierarchical}
Q.~T. Thai, J.-C. Yim, T.-W. Yoo, H.-K. Yoo, J.-Y. Kwak, and S.-M. Kim,
  ``Hierarchical byzantine fault-tolerance protocol for permissioned blockchain
  systems,'' \emph{The Journal of Supercomputing}, vol.~75, no.~11, pp.
  7337--7365, 2019.

\bibitem{gupta2020resilientdb}
S.~Gupta, S.~Rahnama, J.~Hellings, and M.~Sadoghi, ``Resilientdb: Global scale
  resilient blockchain fabric,'' \emph{Proc. {VLDB} Endow.}, vol.~13, no.~6,
  pp. 868--883, 2020.

\bibitem{jiang2023scalable}
W.~Jiang, X.~Wu, M.~Song, J.~Qin, and Z.~Jia, ``A scalable byzantine fault
  tolerance algorithm based on a tree topology network,'' \emph{IEEE Access},
  2023.

\bibitem{zhang2018research}
L.~Zhang and Q.~Li, ``Research on consensus efficiency based on practical
  byzantine fault tolerance,'' in \emph{2018 10th international conference on
  modelling, identification and control (ICMIC)}.\hskip 1em plus 0.5em minus
  0.4em\relax IEEE, 2018, pp. 1--6.

\bibitem{feng2018scalable}
L.~Feng, H.~Zhang, Y.~Chen, and L.~Lou, ``Scalable dynamic multi-agent
  practical byzantine fault-tolerant consensus in permissioned blockchain,''
  \emph{Applied Sciences}, vol.~8, no.~10, p. 1919, 2018.

\bibitem{qushtom2022high}
H.~Qushtom, J.~Mi{\v{s}}i{\'c}, V.~B. Mi{\v{s}}i{\'c}, and X.~Chang, ``A high
  performance two-layer consensus architecture for blockchain-based {IoT}
  systems,'' \emph{Peer-to-Peer Networking and Applications}, pp. 1--13, 2022.

\bibitem{li2021optimized}
Y.~Li, L.~Qiao, and Z.~Lv, ``An optimized byzantine fault tolerance algorithm
  for consortium blockchain,'' \emph{Peer-to-Peer Networking and Applications},
  vol.~14, pp. 2826--2839, 2021.

\bibitem{veronese2011efficient}
G.~S. Veronese, M.~Correia, A.~N. Bessani, L.~C. Lung, and P.~Verissimo,
  ``Efficient byzantine fault-tolerance,'' \emph{IEEE Transactions on
  Computers}, vol.~62, no.~1, pp. 16--30, 2011.

\bibitem{decouchant2022damysus}
J.~Decouchant, D.~Kozhaya, V.~Rahli, and J.~Yu, ``{DAMYSUS}: streamlined {BFT}
  consensus leveraging trusted components,'' in \emph{Proceedings of the 17th
  European Conference on Computer Systems}, 2022, pp. 1--16.

\bibitem{10.1145/3552326.3587455}
S.~Gupta, S.~Rahnama, S.~Pandey, N.~Crooks, and M.~Sadoghi, ``Dissecting bft
  consensus: In trusted components we trust!'' in \emph{Proceedings of the 18th
  European Conference on Computer Systems}, ser. EuroSys '23, New York, NY,
  USA, 2023, p. 521–539.

\bibitem{bessani2023vivisecting}
A.~Bessani, M.~Correia, T.~Distler, R.~Kapitza, P.~Esteves-Verissimo, and
  J.~Yu, ``Vivisecting the dissection: On the role of trusted components in bft
  protocols,'' 2023.

\bibitem{yin2019hotstuff}
M.~Yin, D.~Malkhi, M.~K. Reiter, G.~G. Gueta, and I.~Abraham, ``{HotStuff}:
  {BFT} consensus with linearity and responsiveness,'' in \emph{Proceedings of
  the 2019 ACM Symposium on Principles of Distributed Computing}, 2019, pp.
  347--356.

\bibitem{castro2002practical}
M.~Castro and B.~Liskov, ``Practical byzantine fault tolerance and proactive
  recovery,'' \emph{ACM Transactions on Computer Systems (TOCS)}, vol.~20,
  no.~4, pp. 398--461, 2002.

\bibitem{bessani2014state}
A.~Bessani, J.~Sousa, and E.~E. Alchieri, ``State machine replication for the
  masses with {BFT}-smart,'' in \emph{2014 44th Annual IEEE/IFIP International
  Conference on Dependable Systems and Networks}.\hskip 1em plus 0.5em minus
  0.4em\relax IEEE, 2014, pp. 355--362.

\bibitem{kotla2007zyzzyva}
R.~Kotla, L.~Alvisi, M.~Dahlin, A.~Clement, and E.~Wong, ``Zyzzyva: speculative
  byzantine fault tolerance,'' in \emph{Proceedings of twenty-first ACM SIGOPS
  symposium on Operating systems principles}, 2007, pp. 45--58.

\bibitem{haeberlen2007peerreview}
A.~Haeberlen, P.~Kouznetsov, and P.~Druschel, ``Peerreview: Practical
  accountability for distributed systems,'' \emph{ACM SIGOPS operating systems
  review}, vol.~41, no.~6, pp. 175--188, 2007.

\bibitem{Chun+Maniatis+Shenker+Kubiatowicz:sosp:2007}
B.~Chun, P.~Maniatis, S.~Shenker, and J.~Kubiatowicz, ``Attested append-only
  memory: making adversaries stick to their word,'' 2007, pp. 189--204.

\bibitem{Levin+Douceur+Lorch+Moscibroda:usenix:2009}
D.~Levin, J.~R. Douceur, J.~R. Lorch, and T.~Moscibroda, ``{TrInc}: Small
  trusted hardware for large distributed systems,'' 2009, pp. 1--14.

\bibitem{xu2021concurrent}
X.~Xu, D.~Zhu, X.~Yang, S.~Wang, L.~Qi, and W.~Dou, ``Concurrent practical
  byzantine fault tolerance for integration of blockchain and supply chain,''
  \emph{ACM Transactions on Internet Technology (TOIT)}, vol.~21, no.~1, pp.
  1--17, 2021.

\bibitem{wang2020beh}
L.-e. Wang, Y.~Bai, Q.~Jiang, V.~C. Leung, W.~Cai, and X.~Li, ``Beh-raft-chain:
  a behavior-based fast blockchain protocol for complex networks,'' \emph{IEEE
  Transactions on Network Science and Engineering}, vol.~8, no.~2, pp.
  1154--1166, 2020.

\bibitem{wen2020dp}
F.~Wen, L.~Yang, W.~Cai, and P.~Zhou, ``Dp-hybrid: a two-layer consensus
  protocol for high scalability in permissioned blockchain,'' in
  \emph{Blockchain and Trustworthy Systems: Second International Conference,
  BlockSys 2020, Dali, China, August 6--7, 2020, Revised Selected Papers
  2}.\hskip 1em plus 0.5em minus 0.4em\relax Springer, 2020, pp. 57--71.

\bibitem{li2020scalable}
W.~Li, C.~Feng, L.~Zhang, H.~Xu, B.~Cao, and M.~A. Imran, ``A scalable
  multi-layer {PBFT} consensus for blockchain,'' \emph{IEEE Transactions on
  Parallel and Distributed Systems}, vol.~32, no.~5, pp. 1146--1160, 2020.

\bibitem{javad2021saguaro}
M.~Javad~Amiri, Z.~Lai, L.~Patel, B.~Thau~Loo, E.~Lo, and W.~Zhou, ``Saguaro:
  Efficient processing of transactions in wide area networks using a
  hierarchical permissioned blockchain,'' \emph{arXiv e-prints}, pp.
  arXiv--2101, 2021.

\bibitem{danezis2022narwhal}
G.~Danezis, L.~Kokoris-Kogias, A.~Sonnino, and A.~Spiegelman, ``Narwhal and
  tusk: a dag-based mempool and efficient {BFT} consensus,'' in
  \emph{Proceedings of the Seventeenth European Conference on Computer
  Systems}, 2022, pp. 34--50.

\bibitem{spiegelman2022bullshark}
A.~Spiegelman, N.~Giridharan, A.~Sonnino, and L.~Kokoris-Kogias, ``Bullshark:
  Dag {BFT} protocols made practical,'' in \emph{Proceedings of the 2022 ACM
  SIGSAC Conference on Computer and Communications Security}, 2022, pp.
  2705--2718.

\bibitem{baird2016swirlds}
L.~Baird, ``The swirlds hashgraph consensus algorithm: Fair, fast, byzantine
  fault tolerance,'' \emph{Swirlds Tech Reports SWIRLDS-TR-2016-01, Tech. Rep},
  vol.~34, 2016.

\bibitem{gkagol2019aleph}
A.~Gagol, D.~Lesniak, D.~Straszak, and M.~Swietek, ``Aleph: Efficient atomic
  broadcast in asynchronous networks with byzantine nodes,'' in
  \emph{Proceedings of the 1st ACM Conference on Advances in Financial
  Technologies}, 2019, pp. 214--228.

\bibitem{sompolinsky2015secure}
Y.~Sompolinsky and A.~Zohar, ``Secure high-rate transaction processing in
  bitcoin,'' in \emph{International conference on financial cryptography and
  data security}.\hskip 1em plus 0.5em minus 0.4em\relax Springer, 2015, pp.
  507--527.

\bibitem{miller2016honey}
A.~Miller, Y.~Xia, K.~Croman, E.~Shi, and D.~Song, ``The honey badger of {BFT}
  protocols,'' in \emph{Proceedings of the 2016 ACM SIGSAC conference on
  computer and communications security}, 2016, pp. 31--42.

\bibitem{crain2021red}
T.~Crain, C.~Natoli, and V.~Gramoli, ``Red belly: a secure, fair and scalable
  open blockchain,'' in \emph{2021 IEEE Symposium on Security and Privacy
  (SP)}.\hskip 1em plus 0.5em minus 0.4em\relax IEEE, 2021, pp. 466--483.

\bibitem{stathakopoulou2019mir}
C.~Stathakopoulou, T.~David, and M.~Vukolic, ``Mir-{BFT}: High-throughput {BFT}
  for blockchains,'' \emph{arXiv preprint arXiv:1906.05552}, 2019.

\bibitem{korkmaz2022alder}
K.~Korkmaz, J.~Bruneau-Queyreix, S.~B. Mokhtar, and L.~R{\'e}veill{\`e}re,
  ``Alder: Unlocking blockchain performance by multiplexing consensus
  protocols,'' in \emph{2022 IEEE 21st International Symposium on Network
  Computing and Applications (NCA)}, vol.~21.\hskip 1em plus 0.5em minus
  0.4em\relax IEEE, 2022, pp. 9--18.

\bibitem{gilad2017algorand}
Y.~Gilad, R.~Hemo, S.~Micali, G.~Vlachos, and N.~Zeldovich, ``Algorand: Scaling
  byzantine agreements for cryptocurrencies,'' in \emph{Proceedings of the 26th
  symposium on operating systems principles}, 2017, pp. 51--68.

\bibitem{luu2016secure}
L.~Luu, V.~Narayanan, C.~Zheng, K.~Baweja, S.~Gilbert, and P.~Saxena, ``A
  secure sharding protocol for open blockchains,'' in \emph{Proceedings of the
  2016 ACM SIGSAC conference on computer and communications security}, 2016,
  pp. 17--30.

\bibitem{zamani2018rapidchain}
M.~Zamani, M.~Movahedi, and M.~Raykova, ``Rapidchain: Scaling blockchain via
  full sharding,'' in \emph{CCS}, 2018.

\bibitem{dwork1988consensus}
C.~Dwork, N.~Lynch, and L.~Stockmeyer, ``Consensus in the presence of partial
  synchrony,'' \emph{Journal of the ACM (JACM)}, vol.~35, no.~2, pp. 288--323,
  1988.

\bibitem{bracha1987asynchronous}
G.~Bracha, ``Asynchronous byzantine agreement protocols,'' \emph{Information
  and Computation}, vol.~75, no.~2, pp. 130--143, 1987.

\bibitem{bonomi2021practical}
S.~Bonomi, J.~Decouchant, G.~Farina, V.~Rahli, and S.~Tixeuil, ``Practical
  byzantine reliable broadcast on partially connected networks,'' in \emph{2021
  IEEE 41st International Conference on Distributed Computing Systems
  (ICDCS)}.\hskip 1em plus 0.5em minus 0.4em\relax IEEE, 2021, pp. 506--516.

\bibitem{shamir1979share}
A.~Shamir, ``How to share a secret,'' \emph{Communications of the ACM},
  vol.~22, no.~11, pp. 612--613, 1979.

\bibitem{vassantlal2022cobra}
R.~Vassantlal, E.~Alchieri, B.~Ferreira, and A.~Bessani, ``Cobra: Dynamic
  proactive secret sharing for confidential {BFT} services,'' in \emph{2022
  IEEE symposium on security and privacy (SP)}.\hskip 1em plus 0.5em minus
  0.4em\relax IEEE, 2022, pp. 1335--1353.

\end{thebibliography}

% \section{Max 1pg Refs}

\end{document}